\documentclass[conference,a4paper]{IEEEtran}

\usepackage{amsmath}
\usepackage{cases}
\usepackage{amsfonts}
\usepackage{amssymb}
\usepackage{boxedminipage}
\usepackage{graphicx}
\usepackage[usenames]{color}
\usepackage{array,color}
\usepackage{colortbl}
\usepackage{extarrows}
\usepackage{bm}
\usepackage{amsmath}
\usepackage[numbers,sort&compress]{natbib}

\newtheorem{theorem}{Theorem}
\newtheorem{lemma}{Lemma}

\begin{document}
\baselineskip=12pt
\title{Queueing Analysis for Block Fading Rayleigh Channels in the Low SNR Regime}

\author{Yunquan~Dong$^*$~and~Pingyi~Fan$^{*\dag}$~\IEEEmembership{Senior Member,~IEEE}\\
\{dongyq08@mails,~fpy@mail\}.tsinghua.edu.cn\\
$^{*}$Department of Electronic Engineering, Tsinghua University, Beijing, 100084, China;\\
$^{\dag}$National Mobile Communication Research Laboratory, Southeast University, China, 210096.
}

\maketitle

\begin{abstract}

Wireless fading channels suffer from both channel fadings and Additive White Gaussian Noise (AWGN). 
    As a result, it is impossible for fading channels to support a constant rate data stream without using buffers. 
In this paper, we consider information transmission over an infinite-buffer-aided block Rayleigh fading channel in the low signal-to-noise ratio (SNR) regime.
    We characterize the transmission capability of the channel in terms of stationary queue length distribution, packet delay, as well as data rate. 
Based on the memoryless property of the service provided by the channel in each block,  we formulate the transmission process as a discrete time discrete state $D/G/1$ queueing problem.
    The obtained results provide a full characterization of block Rayleigh fading channels and can be extended to the finite-buffer-aided transmissions.

\end{abstract}

\begin{keywords}
block Rayleigh fading channel, buffer-aided communication, queueing analysis, queue length distribution, packet delay.
\end{keywords}

\section{Introduction}
Wireless communication have found more and more applications in recent years, such as 4G cellular network, WLANs, satellite communication and high speed railway communications. Unlike the traditional wireline communications, wireless communications suffer from channel fading besides the Gaussian noise, which leads to the fluctuation of  instantaneous channel capacity and brings great difficulty to the evaluation and utilization of wireless channels. Therefore, characterizing \textit{what kind of service a fading channel can provide} is a key challenge in wireless communications.

The transmission capability of a fading channel is often characterized by ergodic capacity \cite{MS-CWB} or outage capacity \cite{GS-Throughput, ADTse-OutageClsnr, LA-Outage}.
    In particular, ergodic capacity is the statistical average value of the instantaneous capacity, which specifies the maximum transmission rate over a fading channel in a large time scale;
 $\epsilon$-outage capacity is the maximum achievable transmission rate under some outage probability constraint.
    It is clear that $\epsilon$-outage capacity focuses on the worst case and the reliability of communications.
 A part of channel service would be wasted during periods when  channel condition is very good.
    Therefore, it is not enough to describe capability of channels only using data rate.
 To this end, the authors in \cite{DN-EC} proposed a link layer channel model termed effective capacity to jointly consider parameters such as traffic rate, queue length/packet delay violations.
    Note that effective capacity is proposed based on large deviation theory and is accurate for some large buffers and delays.

In fact, wireless communications and queueing theory are connected with each other naturally. First, due to the fluctuation of the instantaneous channel capacity, a buffer must be used at the transmitter to match the source traffic with the channel transmission capability, which is a typical queueing problem. Second, buffers are commonly used in engineering practice.
Recently, Gallager and Berry discussed the optimal power allocation policy under some delay constraints in a finite-buffer aided point-to-point wireless communication system \cite{RG-CoF}. In this paper, we will consider the buffer aided communication over a block rayleigh fading channels in the low SNR (signal-to-noise) regime, using tools from  queueing theory. In a nutshell, to characterize what kind of service the fading channels can provide, we examine the channel with a constant rate data input to the buffer, and describe the channel transmission capability from  aspects such as queue length distribution and packet delay.
Although we have assumed that the buffer is infinitely long, the results can readily be extended to the finite-buffer case.

However, there are some challenges in applying queueing theory to this problem.
    Note that the channel gain varies block by block for fading channels. Thus, the time is discrete for fading communications.
Moreover, the channel gain is a positive real number ranging from zero to infinity.
    Therefore, the input and output process of fading channels would be a \textit{discrete time continuous state} Markov process, for which little references exists.
To this end, one has to resort to other techniques such as the stochastic process method \cite{DF_Overflowrate},  transforming the continuous state space to discrete state space using quantization \cite{DF_Discrete}.

In this paper, we developed another method to transform the discrete time continuous states Markov process into a discrete time discrete state Markov chain.
    The key idea behind the formulation is the memoryless property of exponential distributions.
To be specific, in one block, the service provided by a low-SNR Rayleigh fading channel follows negative exponentially distribution.
     Although a part of the service ability of the block has been consumed by a previous packet, the remaining service capability of this block follows the same distribution as itself, which can be seen as the service ability provided by a new block.
Based on this observation, we denote the service time of a packet as the integer part of its actual service time.
    Thus, the state space of the queue length at discrete epochs (the beginning of blocks) will also be discrete.
That is, the original discrete time continuous state queueing process is transformed in to a discrete time discrete state Markov chain.
    Based on this model, we investigate the stationary distribution of queue length and packet delay in closed forms in this paper.

The rest of this paper is organized as follows. The channel model and the  queueing model formulation is presented in Section \ref{sec:2}.
    The stationary queue length distribution is investigated in \ref{sec:3}, where the it is proved that the stationary distribution of the queueing process at the departure epochs and at the arbitrary epochs are the same.
Average packet delay is obtained in Section \ref{sec:4}. The obtained result is also be presented via numerical results in Section \ref{sec:5}.
    Finally, we concluded our work in section \ref{sec:6}.

\section{System Model and Queueing Formulation}\label{sec:2}
\subsection{The Block Fading Channel Model}

Consider a point-to-point communication over a block-fading Rayleigh channel with additive white Gaussian noise (AWGN). For such a channel, the channel gain stays fixed over each block  and varies independently in different blocks. Let $T_B$ be the block length, $g_n$ be the time varying channel gain during the {\em n}-th block and $\gamma_n=g^2_n$ be the  corresponding power gain. Then their probability density functions (\textit{p.d.f.}) are given by, respectively,
\begin{equation}\label{eq:rayle}\nonumber
        p_g(x)=\frac{1}{\sigma^2}e^{\frac{-x^2}{2\sigma^2}},~   p_\gamma(x)=\frac{1}{2\sigma^2}e^{\frac{-x}{2\sigma^2}}.
\end{equation}

Let $P$ denote the transmit power, $W$ is the limited transmit
bandwidth and $N_0$ is the noise power spectral density. Let $\alpha$
and $d$ denote the path loss exponent and the distance between the
transmitter and receiver respectively. Then the instantaneous
capacity of the block Rayleigh fading channel in $nats$ is {
\begin{equation}\label{df:insC}\nonumber
c_n=W\ln\left(1+\frac{\gamma_nPd^{-\alpha}}{WN_0}\right)
\end{equation} }
and the service provided by the fading channel in one block is $s_n=c_nT_B$. Then the amount of service provided by the fading channel in $k$ successive blocks can be expressed by $S_k=\sum_{m=1}^k s_m$. For the convenience of notation, we also define the average received SNR as $\rho=\frac{2\sigma^2P}{WN_0d^\alpha}$.

By some culculation, we can get the cumulative distribution function (CDF) of the instantaneous capacity $c_n$ as $    F_c(x)=1-e^{\frac{-1}{\rho}\left(e^\frac{x}{W}-1\right)}$,
which will reduce to {
\begin{equation}\label{eq:F_cn_x} \nonumber
    F_c(x)=1-e^{\frac{-x}{W\rho}}
\end{equation} }
in the low SNR scenario considered in this paper.

Then the CDF of the service in one block ($s_n$) is given by  {
\begin{equation}\label{eq:F_sn_x}
    F_s(x)=1-e^{\frac{-x}{\nu}},
\end{equation} }
where $\nu= WT_B\rho$. It can be seen that $s_n$ follows the negative exponential distribution and $S_k$ follows the Gamma distribution, whose {\em p.d.f.}
is given by {
\begin{equation}\label{eq:pdf_Sk}\nonumber
    f_{S}(x)=\frac{1}{\Gamma(k) \nu^k}x^{k-1}e^{\frac{-x}{\nu}},
\end{equation} }
where $\Gamma(k)=\int_0^\infty e^{-t}t^{k-1}dt$ is the Gamma function.

\subsection{The Markov Chain Model}

The transmission capability of the time varying fading channel is examined by a constant rate ($R$) data stream. An infinite length First In First Out (FIFO) buffer is used at the transmitter side to match the source traffic stream with the time varying channel service in each block. Let $Q(n)$ be the length of the queue in the buffer in \textit{nats} at the start of block $n$. Assume that the data are served by packets and the packet size equals to the traffic arriving at the buffer in each block, i.e., $L_p=RT_B$. We use $D(n)$ to denote the time that the packet arriving in block $n$ will spend in the queue. It is seen that the queueing process $Q(n)$ is a discrete time continuous state Markov process. Unfortunately, few results are available for such processes. Thus, some transformations are needed to construct a discrete time discrete state $D/G/1$ queue.

Define the service time of a packet $T_n$ as the integer part of its actual service time, i.e., the number of complete blocks of the period in which the packet is served. Firstly, the probability that a packet is served within one block can be derived as follows,  {
\begin{equation}\label{dr:p_0}\nonumber
    p_0=\Pr\{T_n=0\}=\Pr\{ s_n>L_p \} =e^{-\theta},
\end{equation} }
where $\theta=\frac{L_p}{\nu}$.

Similarly, with $F_s(x)$, and $f_{S_k}(x)$, the probability that a packet will be served in $k$ ($k\geq1$) blocks can be obtained as {
\begin{equation}\label{dr:p_k}
\begin{split}
    p_k=&\Pr\{T=k\}\stackrel{(a)}{=}\Pr\{S_{k}\leq L_p, S_{k+1}> L_p\}\\
    \stackrel{(b)}{=}&\int_0^{L_p} f_{S_{k}}(x_{k})dx_{k} \int_{L_p-x_{k}}^\infty f_{s}(x_{k+1})dx_{k+1}\\
    =&\frac{1}{k!}e^{-\theta}\theta^k,
\end{split}
\end{equation} }
where (b) holds because $S_k$ and $s_{k+1}$ are independent from each other. Most importantly, (a) is assumed to hold in the sense of the memoryless property of negative exponential distributions, which is given by the following Lemma.

\begin{lemma}\label{Lem:memoryless} \cite{Fan-SCT}
    If $X$ is an negative exponential distributed random variable, then $X$ is memoryless, namely, {
    \begin{equation}\nonumber
        \Pr\{X\leq s+t|X>s\}=\Pr\{X\leq t\}.
    \end{equation}}
\end{lemma}

Therefore, when we say that the service time of a packet is $T_n=k$, the packet is actually completed in $k+1$ blocks,as shown in (\ref{dr:p_k}.a). In this case, only a part of the service capability of block $n+1$ is consumed. So it can continue the service for the next packet in the queue.  As is shown by (\ref{eq:F_sn_x}), the channel service of each block in the low SNR region follows the negative exponential distribution, which is memoryless by Lemma \ref{Lem:memoryless}. Therefore, the remaining service capability of the {\em n}+1-th block can be seen as the service that can be provided by a new block. In this sense, the packet consumes only $k$ blocks. Similarly, if a packet is completed within one block, it is defined that its service time is zero, i.e., $T_n=0$.

In this way, the problem of information transmission over a block fading Rayleigh channel is transformed into a classical discrete time $D/G/1$ queueing problem. Its arrival process is $\{A_n=1,n\geq1\}$ (unit: $L_p$) and its service time $T_n$ (unite: block) is a Poisson distributed random variable, whose probability generating function (PGF) is given by {
\begin{equation}\label{dr:pgf_t}
    G(z)=\textsf{E}[z^{T_n}]=\sum_{k=0}^\infty p_k z^k=e^{\theta(z-1)}
\end{equation} }
and we have $\textsf{E}[T_n]=G'(t)|_{t=0}=\theta$. Throughout the paper, it is assumed that $\theta<1$ so that the queue is stable.

Up to now, the information transmission over the channel is formulated as an late arrival system queueing model with immediate access, in which each packet arrives at the end of a block ($k^-$) and leaves at the beginning of the block ($k^+$). The service of each packet is also assumed to start at the beginning of a block. If the service time of a packet is zero, i.e., $T_n=0$, it is assumed to leave the queue immediately at the beginning of the block following its arrival.

Let $Q_n^+$ be the number of packets in the queue at the arbitrary time of $n^+$. Usually, it is not a Markov chain.
Define $\tau_n$ be the departure epoch of packet $k$ and $\tau_{k+1}=\tau_k$ if the service time of the {\em k}+1-th packet is zero, i.e. $T_{k+1}=0$. Then $\tau_n^+$ must see either an empty buffer or the start of the service of a new packet. So $\tau_n^+$ is the aftereffectless point of the queue length process. Let $L_n^+=Q(\tau_n^+)$ be the number of packets in the buffer after the departure of the {\em n}-th packet. Then $\{L_n^+,n\geq1\}$ is a Markov chain. Since the arrival process is $\{A_n=1,n\geq1\}$, we have {
\begin{equation}\label{df:L_n_+}\nonumber
    L_{n+1}^+=\left\{ \begin{aligned}
             &L_n^+-1+T_{n+1}, & L_n^+\geq1\\
             &T_{n+1}          & L_n^+=0.
             \end{aligned} \right.
\end{equation} }
and the transition probability matrix of $\{L_n^+,n\geq1\}$ is {
\begin{equation}\label{df:P_matri}\nonumber
    \textbf{P}=\left[
  \begin{array}{ccccccccccccccccccccc}
    p_0 & p_1 & p_2  &\cdots \\
    p_0 & p_1 & p_2  &\cdots \\
    0   & p_0 & p_1  &\cdots \\
    \vdots & \vdots & \vdots & \ddots \\
  \end{array}
\right]
\end{equation} }

\section{Queue Length Distribution}\label{sec:3}
In this section, we will perform a detailed investigation on the queue length process in terms of stationary queue length distribution.

\subsection{The Stationary Queue Length at the Departure Epochs}

Let $L^+=\lim_{n\rightarrow\infty}L_n^+$ be the limitation of the queue length process and $\pi_j=\Pr\{L^+=j\},j\geq0$,
then the vector $\vec{\pi}=\{\pi_0,\pi_1,\cdots\}$ is the stationary queue length at the departure epochs.
\begin{theorem}\label{th:sta_distri_dep}
    If $\theta<1$, the PGF of the stationary queue length at the departure epochs is given by:
    \begin{equation}\label{rt:sta_distri_dep}
        L^+(z)=\frac{ (1-\theta)(1-z) }{ 1-ze^{\theta(1-z)} },
    \end{equation}
    where $\theta=\frac{L_p}{\nu}$.
\end{theorem}

\begin{proof}
    When $\theta<1$, the queue is stable. According to the classical queueing theory,  $\overrightarrow{\pi}\textbf{P}=\overrightarrow{\pi}$ holds. So we have {
    \begin{equation}\nonumber
        \pi_j=\pi_0p_j+\sum_{i=1}^{j+1}\pi_i p_{j+1-i},~~~j\geq0.
    \end{equation} }

    Multiplying $z^j$ on the both sides of the equation and take the sum, one has {
    \begin{equation}\nonumber
        \begin{split}
            L^+(z)=&\pi_0\sum_{j=0}^\infty p_j z^j + \sum_{j=0}^\infty z^j \sum_{i=1}^{j+1}\pi_i p_{j-i+1}\\
                  =&\pi_0G(z)+\frac{1}{z}[L^+(z)-\pi_0]G(z),
        \end{split}
    \end{equation} }
    where $G(z)$ is the PGF of the service time $T_n$ given by (\ref{dr:pgf_t}).

    Solving $L^+(z)$ from above equation, one can get {
    \begin{equation}\label{eq:L_z}
        L^+(z)=\frac{ \pi_0(1-z)G(z) }{ G(z)-z }=\frac{ \pi_0(1-z) }{ 1-ze^{\theta(1-z)} }.
    \end{equation} }
    By the property of PGF, we have {
    \begin{equation}\nonumber
            1=\lim_{z\rightarrow1}L^+(z)=\frac{\pi_0}{1-\theta}.
    \end{equation} }

    Thus, $\pi_0=1-\theta$. With $\pi_0$ and (\ref{eq:L_z}), \textit{Theorem} \ref{th:sta_distri_dep} is proved.


\end{proof}

\subsection{The Stationary Queue Length at the Arbitrary Epochs}
This subsection deals with the stationary queue length at the arbitrary epochs. Firstly, let's introduce two Lemmas that will be used.

    \begin{lemma}\label{lem:mean}
        $X$ is a discrete random variable of non-negative integers and $\Pr\{X=k\}=p_k$, then
        \begin{equation}\nonumber
            \textsf{E}[X]=\sum_{k=0}^\infty \Pr\{X>k\}.
        \end{equation}
    \end{lemma}

    The proof of this lemma is as follows {
        \begin{equation}\nonumber
            \textsf{E}[X]=\sum_{l=1}^\infty \sum_{k=1}^l p_l
            =\sum_{k=1}^\infty \sum_{l=k}^\infty p_l=\sum_{k=0}^\infty \Pr\{X>k\}.
        \end{equation} }

\begin{lemma}\label{lem:pro_pgf}
    $X$ is a discrete random variable of non-negative integers and $G(z)$ is its PGF, then
    \begin{equation}\nonumber
        \sum_{k=0}^\infty z^k \Pr\{X>k\}=\frac{1-G(z)}{1-z}.
    \end{equation}
\end{lemma}
\begin{proof} {\small
    \begin{equation}\nonumber
        \begin{split}
            &\sum_{k=0}^\infty z^k \Pr\{X>k\}
            =\sum_{k=0}^\infty z^k \sum_{j=k+1}^k \Pr\{X=j\}\\
            =&\sum_{j=1}^\infty \sum_{k=0}^{j-1} z^k \frac{z-1}{z-1} \Pr\{X=j\}\\
            =&\sum_{j=1}^\infty \frac{z^k-1}{z-1} \Pr\{X=j\}
            =\frac{1-G(z)}{1-z}.
        \end{split}
    \end{equation} }
\end{proof}

Next, the stationary queue length at the arbitrary epochs is specified by the following theorem.
\begin{theorem}\label{th:sta_distri_arbtr}
    The stationary queue length distribution at the departure epochs and arbitrary epochs are the same.
\end{theorem}

\begin{proof}
    Denote the limit distribution of the queue length process at arbitrary epochs as $v_i=\lim_{n\rightarrow\infty}\Pr\{Q_n^+=i\}$ for $i\geq0$. Then Theorem \ref{th:sta_distri_arbtr} will be proved if $v_i=\pi_i$ or $V(z)=L^+(z)$.

    Let $\widehat{L}_n$ be the number of packets in the buffer at the last departure epoch before $n^+$. It stays  unchanged at every boundary points of blocks between two adjacent departure epochs, no matter how much packets arrives during this period. Define $u_i=\lim_{n\rightarrow\infty}\Pr\{\widetilde{L}_n=i\}$.

    Let $\Delta \tau_i$ be the sojourn time that $\widetilde{L}_n$ stays at state $i$ and $m_i=\textsf{E}[\Delta \tau_i]$. When the buffer is non-empty, the sojourn time equals to the service time of the head of line packet in the buffer. However, when the buffer is empty, it must wait one block for the arrival of the next packet first. So we have {\small
    \begin{equation}\nonumber
        \Delta \tau_i=\left\{ \begin{aligned}
             &1+T & i=0\\
             &T &   i\geq1
             \end{aligned} \right.~\mbox{and}~
        m_i=\left\{ \begin{aligned}
             &1+\theta &i=0\\
             &\theta &i\geq1.
             \end{aligned} \right.
    \end{equation} }

    Then the average stationary sojourn time of $\{\widetilde{L}_n,n\geq0\}$ is $\overline{m}=\sum_{j=0}^\infty \pi_j m_j=1$. In this sense, $\frac{\pi_i m_i}{\overline{m}}$ is a probability distribution and $u_i=\frac{\pi_i m_i}{\overline{m}}=\pi_i m_i$ holds according to the theory of renewal process.

    Define $\eta_k$ as the elapsed sojourn time at state $k$ until $n^+$. Since $m_k=\textsf{E}[\Delta\tau_k]=\sum_{l=0}^\infty \Pr\{\Delta\tau_k>l\}$, we know that $\Pr\{\eta_k=l\}=\frac{1}{m_k}\Pr\{\Delta \tau_k>l\}$  is a distribution law, which equals to $\frac{1}{1+\theta}\Pr\{T>l-1\} ~\mbox{for}~ k=0$
    and {    $\frac{1}{\theta}\Pr\{T>l\} ~\mbox{for}~ k\geq1$}

    Assume there are $k$ packets left in the buffer after the departure of the last packet before $n^+$, i.e., $\widetilde{L}_n=k$. Recall that $Q_n^+$ is the queue length at time $n^+$. Then $Q_n^+=j$ means that there are $j-k$ packets arrived during $\eta_k$. We have {\small
    \begin{equation}\nonumber
    \begin{split}
            v_j=&\sum_{k=0}^j u_k \Pr\{\eta_k=j-k\}\\
            =&\pi_0 \Pr\{T>j-1\}+\sum_{k=1}^j \pi_k \Pr\{T>j-k\}.
    \end{split}
    \end{equation} }
    Particularly, $v_0=\frac{u_0}{m_0}\Pr\{\Delta\tau_0>0\}=\pi_0$.

    Multiplying $z^j$ on the both sides and taking the sum, one can get{\small
    \begin{equation}\nonumber
        \begin{split}
            V(z)&=\pi_0+\sum_{j=1}^\infty \pi_0\Pr\{T>j-1\} z^j + \sum_{j=1}^\infty z^j \sum_{k=1}^j \pi_k \Pr\{T>j-k\}\\
            &\stackrel{(a)}{=}\pi_0+\pi_0 z\frac{1-G(z)}{1-z} + [L^+(z)-\pi_0]\frac{1-G(z)}{1-z}\\
            &=\frac{\pi_0G(z)(1-z)}{G(z)-z}
            =L^+(z),
        \end{split}
    \end{equation} }
    where (a) follows \textit{Lemma} \ref{lem:pro_pgf}. This completes the proof.

\end{proof}

\subsection{Stationary Queue Length Distribution}
The stationary distribution of the queue length process can be obtained from its PGF (\ref{rt:sta_distri_dep}), \textit{Theorem} \ref{th:sta_distri_dep}. For $k=0$, we know that $\pi_0=1-\theta$.

Before the following discussion, let's introduce one lemma that will be used.
\begin{lemma}\label{lem:causys_highorder}
    \textit{Cauchy Integral Formula} (\textit{extended}) \cite{Xijiao-Fubian}. Let $C$ be a simple closed positively oriented piecewise smooth curve on a domain $D$, and let the function $f(z)$ be analytic in a neighborhood of $C$ and its interior. Then for every $z_0$ in the interior of $C$ and every natural number $n$, we have that $f^{(n)}(z)$ is n-times differentiable at $z_0$ and its derivative is{\small
    \begin{equation}\label{eq:cauchy's_inti}
        f^{(n)}(z_0)=\frac{n!}{2\pi i}\oint_C \frac{f(z)}{(z-z_0)^{n+1}} dz,~n=1,2,\cdots.
    \end{equation} }
\end{lemma}

Let $C$ and $C'$ are both circles centered on the origin with their radiuses $r<1$ and $1<r'<\frac{-1}{\theta}W_{-1}(-\theta e^{-\theta})$.

Since $L^+(z)$ is convergent within the unit circle, for $k \geq 1$, the stationary can be expressed by{\small
\begin{equation}\label{dr:pi_infty}\nonumber
    \pi_k=\frac{1}{2\pi i}\oint_{C} L^+(z) \frac{dz}{z^{k+1}}
    =\frac{1-\theta}{2\pi i} \oint_{C} \frac{1-z}{1-ze^{\theta(1-z)}} \frac{dz}{z^{k+1}}.
\end{equation} }

It can be seen that function $g(z)=\frac{1-z}{1-ze^{\theta(1-z)}} \frac{1}{z^{k+1}}$ has two singular points, namely $z=1$ and $z=\frac{-1}{\theta}\textsf{W}_{-1}(-\theta e^{-\theta})$, where $\textsf{W}_{-1}(z)$ is the lower branch Lambert W function. Particularly, $z=1$ is one movable singularity since $\lim_{z\rightarrow1}g(z)=1$ is finite. Therefore, within circle $C'$, $g(z)$ can be considered as analytic. According to Cauchy-Goursat's theory, the integral of $g(z)$ along any closed curve in $C'$ is a invariable. Then we have {\small
\begin{equation}\label{dr:pi_infty1}\nonumber
\begin{split}
    \pi_k=&\frac{1-\theta}{2\pi i} \oint_{C'} \frac{(z-1)\frac{e^{\theta(z-1)}}{z}  }{\frac{e^{\theta(z-1)}}{z}-1} \frac{dz}{z^{k+1}}\\
    \stackrel{(a)}{=}&\frac{1-\theta}{2\pi i} \oint_{C'} (z-1)\frac{e^{\theta(z-1)}}{z} \sum_{j=0}^\infty \left(\frac{e^{\theta(z-1)}}{z}\right)^j \frac{dz}{z^{k+1}}\\
    \stackrel{(b)}{=}&(1-\theta) \sum_{j=1}^\infty \left[ \frac{1}{(k+j-1)!}(j\theta)^{k+j-1} - \frac{1}{(k+j)!}(j\theta)^{k+j}\right] e^{-j\theta}
\end{split}
\end{equation} }
where (a) follows $\frac{1}{1-z}=\sum_{j=0}^\infty z^j$ for $|z|<1$ and $\frac{e^{\theta(z-1)}}{z}<1$ on circle $C'$, (b) follows \textit{Lemma} \ref{lem:causys_highorder} and $e^{j\theta(z-1)}$ is analytic.

Define $\varphi_{-1}=1$ and {\small
\begin{equation}\label{rt:fai_0}\nonumber
    \varphi_k=(1-\theta)\sum_{j=1}^\infty \frac{1}{(k+j)!}(j\theta)^{k+j} e^{-j\theta},
\end{equation} }
for $k\geq0$, the stationary queue length distribution turns to be {\small
\begin{equation}\label{rt:pij}\nonumber
    \pi_k=\varphi_{k-1}-\varphi_k.
\end{equation} }

\section{Packet Delay}\label{sec:4}

The \textit{packet delay} $D$ is defined as the time interval between the arrival of a packet and its departure. Firstly, the packet delay consists of a service time $T$. If the packet arrives seeing a non-empty buffer, it must wait for a \textit{waiting time} $W$ for its service. Finally, for the formulation in this paper, there is another piece of time that the packet spends in the system, i.e., the \textit{vestige time} $V$. In this paper, $T=k$ means that the service of a packet is not finished until the $k+1$-th block. According to the memoryless property of negative exponential distribution, although part of the service ability of the $k+1$-th block has been consumed, it is considered as a brand new block. However, the packet still has to spend a part of that block in the system, which is called the \textit{vestige time}. Thus, we know that{
\begin{equation}\label{df:delay}
    D=T+W+V.
\end{equation} }
\begin{theorem}\label{th:delay}
    In the FIFO discipline, the average packet delay is given by
        \begin{equation}\label{rt:E_D}\nonumber
            \textsf{E}[D]=\frac{1}{2}+\theta+\frac{\theta^2}{2(1-\theta)}+\int_0^1 (x-1) e^{\frac{-\theta}{x}} dx.
        \end{equation}
\end{theorem}
\begin{proof}
    Since there is only one packet arrives in each block, the average time that a packet stays in the system ($S+W$) equals to the average queue length by the Little's law. So we have {
    \begin{equation}\label{dr:sw}
        \textsf{E}[S+W]=\textsf{E}[L^+]=\lim_{z\rightarrow1}L^+(z)=\frac{\theta(2-\theta)}{2(1-\theta)}.
    \end{equation}}

    To investigate the vestige time $V$, we have to define some new random variables. Firstly, for $k\geq1$, define {
    \begin{equation}\nonumber
        \begin{split}
            &U_k=L_p-S_k,~~~U_k^+=U_k|_{U_k>0},\\
            &V_k=\frac{U_k^+}{s_n},~~~~~~~~V_k^-=V_k|_{V_k<1},
        \end{split}
    \end{equation} }
    where $s_n$ is the service provided by the channel in one block and $S_k$ is the total amount of service provided by $k$ successive blocks.
    By the definition, $U_k$ is the remaining part of packet after $k$ blocks of transmission, if it is positive. Using a positive condition, we get $U_k^+$. Assuming the channel service of the next block is $s_n$, the remaining amount of data $U_k^+$ needs $V_k$ blocks to be transmitted. However, $V_k$ is the vestige time only if the remaining data can be transmitted within this block, which turns to $V_k^+$.

    As we can see, $U_k\in(-\infty,L_p)$, $U_k^+\in(0,L_p)$, $V_k\in(0,\infty)$ and $V_k^+\in(0,1)$. Particularly, the CDF of $U_k$ is {
    \begin{equation}\nonumber
        \begin{split}
        F_{U_k}(x)=&\Pr\{U_k\leq x\}=\Pr\{S_k\geq L_p-x\}\\
        =&\frac{1}{\Gamma(k)}\int_{\frac{L_p-x}{\nu}}^\infty t^{k-1}e^{-t} dt,
        \end{split}
    \end{equation} }
    where $\Gamma(k,x)=\frac{1}{\Gamma(k)}\int_x^\infty t^{k-1}e^{-t} dt$ is the upper Gamma function and $\nu=WT_B\rho$.

    Therefore, the CDF of $U^+$ is{
    \begin{equation}\nonumber
        \begin{split}
            F_{U_k^+}(x)=&\Pr\{U_k^+\leq x|U_k^+>0\}\\
            =&\frac{\int_{\frac{L_p-x}{\nu}}^\infty t^{k-1} e^{-t} dt - \Gamma(k, \theta)}   {\gamma(k,\theta )},
        \end{split}
    \end{equation} }
    where $\gamma(k,x)=\frac{1}{\Gamma(k)}\int_0^x t^{k-1}e^{-t} dt$ is the lower Gamma function.


    Since $U_k^+$ and the channel service of the next block ($s_n$) are random variables independent from each other, the CDF of $V_k$ can be derived as follows.{\small
    \begin{equation}\nonumber
        \begin{split}
        F_{V_k}(x)=&\Pr\{V_k\leq x \}=\Pr\{U_k^+\leq xs_n \}\\
        =&\iint_D dF_{U_k^+}(u)dF_s(s)\\
        =&\frac{1}{\nu^k\gamma(k,\theta)} \int_0^{L_p} (L_p-u)^{k-1}e^{-\frac{1}{\nu}{(L_p-u+\frac{u}{x})}} du.
        \end{split}
    \end{equation}}

    Next, the CDF of $V_k^-$ is{
    \begin{equation}\nonumber
        \begin{split}
            F_{V_k^-}(x)=&\Pr\{V_k^-\leq x|V_k<1\}=\frac{\Pr\{V_k<x\}}{\Pr\{V_k<1\}}\\
            =&\frac{k}{L_p}\int_0^{L_p} \left(\frac{L_p-u}{L_p}\right)^{k-1}e^{\frac{u}{\nu}{(1-\frac{1}{x})}} du,
        \end{split}
    \end{equation}}
    for $x\in(0,1)$ and the average of $V_k^-$ will be {
    \begin{equation}\label{dr:E_xk_}\nonumber
        \textsf{E}[V_k^-]=1-\frac{k}{L_p}\int_0^1 \int_0^{L_p} \left(\frac{L_p-u}{L_p}\right)^{k-1}e^{\frac{u}{\nu}{(1-\frac{1}{x})}} dxdu,
    \end{equation} }

    For the case of $k=0$, define $V_0=\frac{L_p}{s_n}$ and $V_0^-=V_0|_{V_0<1}$. Specifically, if one packet is completed within one block, the vestige time equals to its actual service time. The CDF of $V_0$ and $V_0^-$ are, respectively{
    \begin{eqnarray*}
      F_{V_0}(x)&=&\Pr\{V_0\leq x\}=e^{-\frac{\theta}{x}}, ~x\in(0,+\infty) \\
      F_{V_0^-}(x)&=&\Pr\{V_0\leq x|V_0<1\}
            =e^{\theta\left(1-\frac{1}{x}\right)}, ~x\in(0,1).
    \end{eqnarray*} }

    Then we can get the expected value of $V_0^-$ {
    \begin{equation}\label{dr:E_V0}\nonumber
       \textsf{E}[V_0^-]=1-e^\theta\int_0^1 e^{-\frac{\theta}{x}}dx.
    \end{equation} }

    Finally, by the whole probability formula, the expected value of vestige time $V$ is {
    \begin{equation}\label{dr:E_V}\nonumber
           \textsf{E}[V]=\textsf{E}[V_0^-]\Pr\{T=0\}+\sum_{k=1}^\infty \textsf{E}[V_k^-]\Pr\{T=k\},
    \end{equation} }
    where {\small
    \begin{equation}\label{dr:E_Vremaining}\nonumber
        \begin{split}
            &\sum_{k=1}^\infty \textsf{E}[V_k^-]\Pr\{T_n=k\}\\
            =&1-e^{-\theta}- \frac{1}{L_p} \int_0^1\int_0^{L_p} \sum_{k=1}^\infty \left(\theta \frac{L_p-u}{L_p}\right)^{k-1}\theta \frac{k}{k!}e^{-\theta} e^{\frac{u}{\nu}{(1-\frac{1}{x})}} dxdu\\
            =&\frac{1}{2}-e^{-\theta}+\int_0^1 x e^{\frac{-\theta}{x}} dx.
        \end{split}
    \end{equation} }

    With this result, we have {
    \begin{equation}\label{rt:E_V}\nonumber
           \textsf{E}[V]=\frac{1}{2}+\int_0^1 (x-1) e^{\frac{-\theta}{x}} dx
    \end{equation} }
    and will complete the proof by combing (\ref{df:delay}) and (\ref{dr:sw}).
\end{proof}

\section{Numerical Results}\label{sec:5}
In this section, we provide some numerical results to illustrate the stationary distribution and packet delay for the infinite buffer model.
The component variance is assumed as $\sigma^2=1$, the system bandwidth is $5$ KHz and transmitting power is $-10$ dBW. Suppose that the distance between the transmitter and the receiver is $1000$ m and the pathloss exponent is $4$. The block length is chosen as $T_B=10^{-4}$ s. According to the definition, we have $\theta=\frac{L_p}{WT_B\rho}$.

\begin{figure}[h]
\centering
\includegraphics[width=3in]{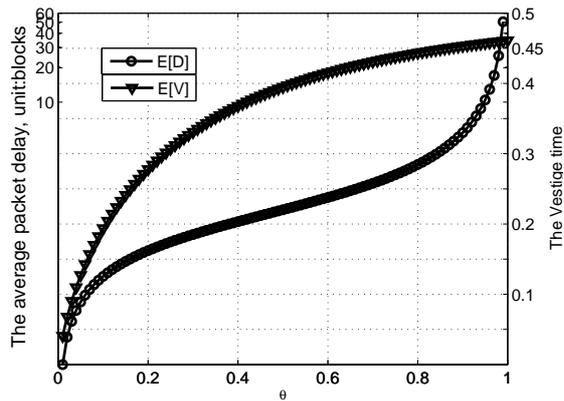}
\caption{The average packet delay of the infinite-buffer model. E[D] corresponds to the log-scale y-axis on the left while E[V] corresponds to the linear y-axis on the right.} \label{fig:ED}
\end{figure}

Besides, we have $L_p=RTB$ and the equivalent AWGN capacity of the channel is $C_a=W\rho$ in the low SNR region. Therefore, $\theta$ can also be seen as the ratio between traffic rate and AWGN capacity.

It is seen in Fig. \ref{fig:ED} that the average delay grows quickly with $\theta$. Particularly, the average delay corresponds to the log-scale y-axis on the left and the average vestige time corresponds to the linear y-axis on the right, both in blocks. It is seen that the average vestige time $\textsf{E}[V]$ is also increasing with $\theta$ but never exceeds 0.5.

The stationary queue length distribution of the infinite-buffer model is shown in Fig. \ref{fig:piinf} for $\theta=0.2,0.5,0.8$, where they are presented in both linear and log-scale y-axis. It is clear that the larger $\theta$ is, the slower $\pi_k$ decreases. This is easy to understand because heavier traffic leads to longer queues. Most importantly, as the queue length  grows, the probability that it occurs decreases approximately exponentially.

\section{Conclusion}\label{sec:6}

Modern communications requires wireless channels to provide QoS  guaranteed services. How to characterize and make full use of the service capability of fading channels is an urgent problem. In this paper, we studied the problem for low-SNR block Rayleigh fading channels by using the memoryless property of block services. However, for general fading channels, the answers are quite unclear, which needs a lot of further efforts.

\section*{Acknowledgement}
This work was supported by the China Major
State Basic Research Development Program (973 Program) No.
2012CB316100(2), National Natural Science Foundation of China
(NSFC) No. 61171064, NSFC No. 61021001 and the open research fund of National Mobile Communication Research Laboratory, Southeast University, China.

\begin{figure}[h]
\centering
\includegraphics[width=3in]{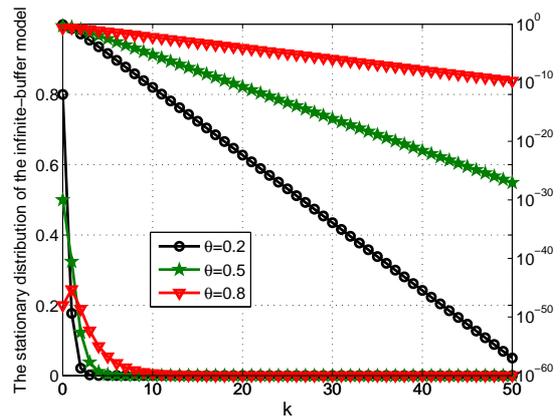}
\caption{The stationary queue length distribution of the infinite-buffer model. Each curve is presented both in linear y-axis on the left and log-scale y-axis on the right.} \label{fig:piinf}
\end{figure}

\bibliographystyle{IEEEtran}

\end{document}